\documentclass[runningheads,envcountsame]{llncs}
\usepackage{amssymb}
\usepackage{graphicx}
\usepackage{url}
\usepackage{color}
\usepackage{xspace}
\usepackage{mathtools}
\usepackage{enumerate}
\usepackage{algorithm}
\usepackage{algorithmicx}
\usepackage{tikz}
\usepackage[color=green!20]{todonotes}
\usepackage{comment}

\newcommand {\uec} {\mathrm{uec}}

\newcommand{\keywords}[1]{\par\addvspace\baselineskip
\noindent\keywordname\enspace\ignorespaces#1}

\usepackage{tikz}
\newcommand{\SV}[1]{#1}
\newcommand{\LV}[1]{}

\SV{\excludecomment{pf}}

\LV{
}

\usepackage{tikz}
\tikzstyle{vertex}=[circle, draw, inner sep=1pt, minimum width=0pt, minimum size=0cm]
\tikzstyle{vertex1}=[circle, draw, inner sep=2pt, minimum width=6pt, minimum size=0.5cm]
\usetikzlibrary{decorations,decorations.pathmorphing,decorations.pathreplacing,fit,positioning}

\title{Weighted Upper Edge Cover: Complexity and Approximability}

\author{Kaveh Khoshkhah\inst{1}\thanks{Supported by the Estonian Research Council, ETAG (Eesti Teadusagentuur), through PUT Exploratory Grant \#620.}\and Mehdi Khosravian Ghadikolaei\inst{2} \and J{\'e}r{\^o}me Monnot\inst{2} \and Florian Sikora\inst{2}
}
\institute{Institute of Computer Science, Tartu university, Estonia
\email{khoshkhah@theory.cs.ut.ee}
\and
Universit\'e Paris-Dauphine, PSL University, CNRS, LAMSADE, 75016 Paris, France
\email{mehdi.khosravian-ghadikolaei,jerome.monnot,florian.sikora@lamsade.dauphine.fr}
}

\titlerunning{Weighted Upper Edge Cover: Complexity and Approximability}
\authorrunning{Khoshkhah \emph{et al.}}
\date{}

\begin{document}
\maketitle
\vspace{-0.5 cm}
\begin{abstract}
Optimization problems consist of either maximizing or minimizing an objective function.
Instead of looking for a maximum solution (resp. minimum solution), one can find a minimum maximal solution (resp. maximum minimal solution).
Such ``flipping" of the objective function was done for many classical optimization problems.
For example, \textsc{Minimum Vertex Cover} becomes \textsc{Maximum Minimal Vertex Cover}, \textsc{Maximum Independent Set} becomes \textsc{Minimum Maximal Independent Set} and so on. In this paper, we propose to study the weighted version of {\em Maximum Minimal Edge Cover} called \textsc{Upper Edge Cover}, a problem having application in genomic sequence alignment.
It is well-known that \textsc{Minimum Edge Cover} is polynomial-time solvable and the "flipped" version is \textbf{NP}-hard, but constant approximable. We show that the weighted \textsc{Upper Edge Cover} is much more difficult
than \textsc{Upper Edge Cover} because it is not $O(\frac{1}{n^{1/2-\varepsilon}})$ approximable, nor $O(\frac{1}{\Delta^{1-\varepsilon}})$ in edge-weighted graphs of size $n$ and maximmm degree $\Delta$ respectively.
Indeed, we give some hardness of approximation results for some special restricted graph classes such as bipartite graphs, split graphs and $k$-trees.
We counter-balance these negative results by giving some positive approximation results in specific graph classes.
\end{abstract}

\vspace{-0.8 cm}

\vspace{-0.4 cm}
\section{Introduction}
\label{intro}

Considering a MaxMin or MinMax version of a problem by ``flipping'' the objective is not a new idea; in fact, such questions have been posed before for many classical optimisation problems.
Some of the most well-known examples include the \textsc{Minimum Maximal Independent Set} problem~\cite{BourgeoisCEP13}
(also known as \textsc{Minimum Independent Dominating Set}), the \textsc{Maximum Minimal Vertex Cover} problem \cite{BoriaCP15}, 
the \textsc{Minimum Maximal Matching} problem (also known as \textsc{Minimum Independent Edge Dominating Set})  \cite{yannakakis1980edge}, 
 and the  \textsc{Maximum Minimal Dominating Set} problem (also called \textsc{Upper Dominating Set}) \cite{AbouEishaHLMRZ16}.
 However, to the best of our knowledge, weighted MaxMin and MinMax versions have not been considered so far, except for \textsc{Minimum Independent Dominating Set} \cite{Chang04,LozinMMZ17},
and \textsc{weighted upper dominating set} problem~\cite{BoyaciM17}.
MaxMin or MinMax versions of classical problems turn out to be much harder than the originals, especially when one considers complexity and approximation.
For example, \textsc{Maximum Minimal Vertex Cover} does not admit any $n^{\frac{1}{2}-\epsilon}$ approximation~\cite{BoriaCP15}, while \textsc{Vertex Cover} admits a simple $2$-approximation.
\textsc{Minimum Maximal Matching} is \textbf{NP}-hard
(but $2$-approximable) while  \textsc{Maximum Matching} is polynomial.

The focus of this paper is on {\em edge cover}.
An  {\em edge cover} of a graph $G=(V,E)$ is a subset of edges $S\subseteq E$ which covers all vertices of $G$.
The {\em edge cover number} of $G=(V,E)$ 
is the minimum size of an {\em edge cover} of $G$. An optimal edge cover can be computed in polynomial time, 
even for the weighted version where a weight is given for each edge and one wants to minimize the sum of the weight of the edges in the solution (called here the {\em weighted edge cover number}).
An edge cover $S\subseteq E$ is {\em minimal} (with respect to inclusion) if the deletion of any subset of edges from $S$ destroys the covering property. Minimal edge cover is also known in the literature as an {\em enclaveless} set~\cite{Slater77} or as a {\em nonblocker} set~\cite{DehneFFPR06}.

In this paper, we study the computational complexity of the {\em weighted upper  edge cover number}, denoted here $\uec(G,w)$, that is the solution with maximum weight among all minimal edge covers. Formally, the associated optimization problem called the \textsc{Weighted Upper Edge Cover} problem asks to find the largest weighted minimal edge cover of an edge-weighted graph.

\vspace{-0.1 cm}
\begin{center}
\fbox{\begin{minipage}{1\textwidth}
\noindent{\textsc{Weighed Upper Edge Cover}}\\\nopagebreak
{\bf Input:} A  weighted connected graph  $G=(V,E,w)$, where $w(e)\geq 0$ for all $e\in E$. \\\nopagebreak
{\bf Solution:} Minimal edge cover $S\subseteq E$. \\\nopagebreak
{\bf Output:} Maximize $w(S)=\sum_{e\in S}w(e)$.
\end{minipage}}
\end{center}
\noindent
Hence, if $S^*$ is an optimal solution of \textsc{Weighed Upper Edge Cover} on $(G,w)$, then  $w(S^*)=\uec(G,w)$.
The unweighted value of the optimal solution is $uec(G)$ (denoted {\em upper  edge cover number}). To the best of our knowledge, the complexity of computing the weighted upper  edge cover number has never been studied in the literature, while a lot of results appear for the unweighed case
(corresponding to $w(e)=1$ for all $e\in E)$ \cite{NguyenSHSMZ08,DBLP:conf/mfcs/AthanassopoulosCKK09,ChenENRRS13,HeL13}.
The unweighted variant was firstly investigated in~\cite{manlove1999algorithmic}, where it is proven that the complexity of computing the upper  edge cover number is equivalent to solve the dominating set problem because $\uec(G)=|V|-\gamma(G)$ where $\gamma(G)$ is the size of minimum dominating set of graph $G$.
We will consider the implications of this important remark afterwards in the paper.

We will now define a related problem useful in the following because it is proved in \cite{manlove1999algorithmic} that $S\subseteq E$ is a minimal edge cover of $G=(V,E)$ iff $S$ is a spanning star forest of $G$ \textit{without trivial stars} (i.e. without stars consisting of a single vertex).

\begin{center}
\fbox{\begin{minipage}{1\textwidth}
\noindent{\textsc{Maximum Weighted Spanning Star Forest problem} (\textsc{MaxWSSF} in short)}\\\nopagebreak
{\bf Input:} An edge-weighted graph $(G,w)$ on $n$ vertices where $G=(V,E)$ and $w(e)\geq 0$ for all $e\in E$. \\\nopagebreak
{\bf Solution:} Spanning star forest  $\mathcal{S}=\{S_1,\dots,S_p\}\subseteq E$. \\\nopagebreak
{\bf Output:} maximizing $w(\mathcal{S})=\sum_{e\in \mathcal{S}}w(e)=\sum_{i=1}^p\sum_{e\in S_i}w(e)$.
\end{minipage}}
\end{center}

Given an instance $(G,w)$ of \textsc{MaxWSSF}, $opt_{MaxWSSF}(G,w)$ denotes the value of an {\em optimal spanning star forest}. Authors of~\cite{NguyenSHSMZ08} describe in details how to apply \textsc{MaxWSSF} model to alignment of multiple genomic sequence, a critical task in comparative genomics. They also show that this approach is promising with real data.
In this model, taking weights into account is fundamental since it represents alignment score. Also, their model uses each edge of the spanning star forest to output the solution.
Therefore, having trivial star is probably undesirable, which enforces the motivation of studying \textsc{Weighed Upper Edge Cover}.

The unweighted version (corresponding to the case $w(e)=1$ for all edges $e$) is denoted by \textsc{MaxSSF}. In this case, the optimal value is $opt_{MaxSSF}(G)$.
For unweighted graphs without isolated vertices, we have $\uec(G)=opt_{MaxSSF}(G)$  since any  spanning star forest (with possible trivial stars) can be (polynomially) converted into a star spanning forest without trivial stars (i.e. a minimal edge cover) with same size \cite{manlove1999algorithmic}. Hence, these two problems are completely equivalent even from an approximation point of view.

Concerning edge-weighted graphs, the relationship between \textsc{Weighed Upper Edge Cover} and \textsc{MaxWSSF} is less obvious.
For instance, we only have the following inequality: $opt_{MaxWSSF}(G,w)\geq \uec(G,w)$ because any minimal edge cover is a particular spanning star forest.
However, the difference between these two values can be arbitrarily large as indicated in Figure~\ref{Fig:exm_wuec,wssf} (in the graph drawn in Figure~ \ref{Fig:exm_wuec,wssf}$.(b)$, $v_4$ is an isolated vertex when  $\varepsilon$ goes to Infinity).
This means that isolated vertices play an important role in feasible solutions.
 Given a spanning star forest $\mathcal{S}=\{S_1,\dots,S_r\}$ of $(G,w)$, we rename vertices such that there is some $p, 0 \leq p< r$ such that $S_i=\{v_i\}$ are trivial stars for all $1\leq i \leq p$ (if $p=0$, then there is no trivial stars), and $S_j$ are non-trivial stars whose $c_j$ is the center for all $j>p$ (if $S_j$ is a single edge, both endpoints are considered as possible centers). We define $\mathrm{Triv}=\{v_i\colon i\leq p\}$ as the set of isolated vertices of $(V,E(\mathcal{S}))$ where $E(\mathcal{S})=\cup_{j>p}^rS_j$; moreover, $V_l$ and $V_c$ are respectively the set of leaves and the set of centers of stars in $V\setminus \mathrm{Triv}$.
Finally, for $v\in V_l$, $e_v(\mathcal{S})=c^\prime v\in E(\mathcal{S})$ denotes the edge linking the center $c^\prime$ to the leaf $v$.

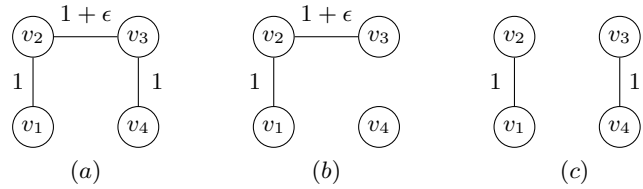
\begin{figure}
\centering
\begin{tikzpicture}[scale=1, transform shape]
\tikzstyle{vertex}=[circle, draw, inner sep=2pt,  minimum width=1 pt, minimum size=0.1cm]
\tikzstyle{vertex1}=[circle, draw, inner sep=2pt, fill=black!100, minimum width=1pt, minimum size=0.1cm]

\node[vertex] (v1) at (-1.2,-1) {$v_1$};
\node[vertex] (v2) at (-1.2,0.2) {$v_2$};
\node[vertex] (v3) at (0.2,0.2) {$v_3$};
\node[vertex] (v4) at (0.2,-1) {$v_4$};
\node () at (-0.5,-1.6) {$(a)$};

\draw (v1)--(v2)--(v3)--(v4);
\node () at (-1.4,-0.4) {$1$};
\node () at (-0.5,0.5) {$1+\epsilon$};
\node () at (0.45,-0.4) {$1$};

\node[vertex] (v11) at (2,-1) {$v_1$};
\node[vertex] (v12) at (2,0.2) {$v_2$};
\node[vertex] (v13) at (3.4,0.2) {$v_3$};
\node[vertex] (v14) at (3.4,-1) {$v_4$};
\node () at (2.7,-1.6) {$(b)$};

\draw (v11)--(v12);
\draw (v13)--(v12);
\node () at (1.75,-0.4) {$1$};
\node () at (2.7,0.5) {$1+\epsilon$};

\node[vertex] (v21) at (5.2,-1) {$v_1$};
\node[vertex] (v22) at (5.2,0.2) {$v_2$};
\node[vertex] (v23) at (6.6,0.2) {$v_3$};
\node[vertex] (v24) at (6.6,-1) {$v_4$};
\node () at (6.0,-1.6) {$(c)$};

\draw (v21)--(v22);
\draw (v23)--(v24);
\node () at (4.95,-0.4) {$1$};
\node () at (6.8,-0.4) {$1$};
\end{tikzpicture}
\caption{$(a):$ The weighted graph $G=(V,E,w)$. $(b):$ Optimal solution of \textsc{MaxWSSF}$(G,w)$. $(c):$ Optimal solution of \textsc{Weighted Upper Edge Cover} for $G$  with value $\uec(G,w)=2$.}\label{Fig:exm_wuec,wssf}
\end{figure}
We mainly focus on specific solutions of \textsc{MaxWSSF} called {\em  nice spanning star forests} defined as follows:

\begin{definition}\label{def:niceSSF}
$\mathcal{S}$ is a nice spanning star forest of $(G,w)$ if
$\mathrm{Triv}=\{v_i\colon i\leq p\}$ is an independent set in $G$ and  all edges of $G$ starting at $\mathrm{Triv}$ are linked to leaves of some $\ell$-stars of $\mathcal{S}$ with $\ell\geq 2$. Moreover, $w(uv)\leq w(e_v(\mathcal{S}))$ for $u\in \mathrm{Triv}$, $v\in V_l$.
\end{definition}
\begin{property}\label{niceSSF}
Any spanning star forest of $(G,w)$ can be polynomially converted into a nice one with at least the same weight.
\end{property}

\begin{proof}
The weights of $(G,w)$ are non-negative. Thus, if $\mathrm{Triv}$ is not an an independent set or if some vertex of $\mathrm{Triv}$ is linked to some center of
$\mathcal{S}$, we could obtain a better spanning star forest with less isolated vertices. In particular, it implies that no vertex of $\mathrm{Triv}$ is linked to a 1-star (i.e. a $K_2$ of $\mathcal{S}$). Finally, if $w(uv)> w(e_v(\mathcal{S}))$, then $\mathcal{S}^\prime=(\mathcal{S}\setminus \{e_v(\mathcal{S})\})\cup \{uv\}$ is a better spanning star forest.
\qed
\end{proof}

It is well known that optimization problems are easier to approximate when the input is a complete weighted graphs satisfying the {\em triangle inequality}, like for example in the traveling salesman problem.
Here, we introduce a generalization of this notion which works to any class of graphs.

\begin{definition}\label{def:cycle-ineq}
An edge weighted graph $(G,w)$ where $G=(V,E)$ satisfies the {\em cycle  inequality}, if for every cycle $C$, we have:
$$
\forall e\in C, ~~2w(e)\leq w(C)=\sum_{e'\in C}w(e')
$$
\end{definition}
Clearly, for complete graphs, cycle and triangle inequality notions coincide.
Definition~\ref{def:cycle-ineq} is interesting when focusing on classes of graphs
like split graphs or $k$-trees. In this article, we are also interested in {\em bivaluate weights} (resp., {\em trivalued})
corresponding to the case $w(e)\in \{a,b\}$ with $0\leq a<b$ (resp., $w(e)\in \{a,b,c\}$ where $0\leq a<b<c$ are 3 reals).
The particular case $a=0$ and $b=1$ (called here {\em binary weights})  is interesting by itself because
\textsc{MaxWSSF}  with binary weights exactly corresponds to  \textsc{MaxSSF} and has been extensively studied in the literature.
Moreover for instance, binary weighted \textsc{Minimum Independent Dominating Set} for chordal graphs has been studied
in \cite{Farber82ORL},  where it is shown that this restriction is polynomial, but bivalued weighted \textsc{Minimum Independent Dominating Set}  for chordal graphs with $a>0$ is \textbf{NP}-hard \cite{Chang04}.

\textbf{\textit{Graph terminology and definitions:}}
Throughout this paper, we consider edge-weighed undirected connected graphs $G=(V,E)$ on $n=|V|$ vertices and $m=|E|$ edges. Each edge $e=uv\in E$ between vertices $u$ and $v$ is weighted by a non-negative weight $w(e)\geq 0$; $K_n$ denotes the {\em complete graph} on $n$ vertices; a {\em bipartite graph} (resp., {\em split graph}) $G=(L\cup R,E)$ is a graph where the vertex set $L\cup R$ is decomposable into an independent set (resp., a clique) $L$ and an independent set $R$. A {\em $k$-tree} is a graph which can be formed by starting from a $k$-clique and then repeatedly adding vertices in such a way that each added vertex has exactly $k$ neighbors completely connected together (this neighborhood is a $k$-clique). For instance, $1$-trees are trees and $2$-trees are maximal series-parallel graphs. A graph is a {\em partial $k$-trees} (or equivalently with {\em treewidth} at most $k$) if it is a subgraph  of a $k$-trees. The {\em degree} $d_G(v)$ of vertex $v\in V$ in $G$ is the number of edges incident to $v$ and $\Delta(G)$ is the {\em maximum degree} of the graph $G$. A {\em star} $S\subseteq E$ of a graph $G=(V,E)$ is a tree of $G$ where at most one vertex has a degree greater than 1, or, equivalently, it is isomorphic to $K_{1,\ell}$ for some $\ell\geq 0$. The vertices of degree 1 (except the center when $\ell\leq 1$) are called {\em leaves} of the star while the remaining vertex is called {\em center} of the star. A $\ell$-star is a star of $\ell$ leaves. If $\ell=0$, the star is called {\em trivial} and it is reduced to a single vertex (the center); otherwise, the star is said {\em non-trivial}.
A {\em  spanning star forest} $\mathcal{S}=\{S_1,\dots,S_p\}\subseteq E$ of $G$ is a spanning forest into stars, that is, each $S_i$ is a star (possibly trivial), $V(S_i)\cap V(S_j)=\emptyset$ and $\cup_{i=1}^p V(S_i)=V$.
An {\em independent set} $S\subseteq V$ of a graph $G=(V,E)$ is a subset of vertices pairwise non-adjacent.
The \textbf{NP}-hard problem \textsc{MaxIS} seeks an independent set of maximum size.
The value of an optimal independent set of $G$ is denoted $\alpha(G)$.
A {\em matching} $M\subseteq E$ is a subset of pairwise non-adjacent edges. A matching $M$ of $G$ is {\em perfect} if all vertices of $G$ are
covered by $M$.
A {\em dominating set} for a graph $G$ is a subset $D$ of $V$ such that every vertex not in $D$ is adjacent to at least one vertex of $D$. The {\em domination number} $\gamma(G)$ is the number of vertices in the smallest dominating set of $G$.

\textbf{\textit{Related work}:}
\textsc{Upper Edge Cover}  has been investigated intensively during the recent years for unweighed graphs, mainly using the terminologies of {\em spanning star forests} or {\em dominating sets}.
The {\em minimum dominating set problem} (denoted \textsc{MinDS}) seeks the smallest dominating set of $G$ of value $\gamma(G)$.
As indicated before, we have $\uec(G)=n-\gamma(G)$.
Thus, using the complexity results known on \textsc{MinDS},
we deduce that \textsc{Upper Edge Cover} is \textbf{NP}-hard in planar graphs of maximum degree 3 \cite{GJ79}, chordal graphs~\cite{BoothJ82} (even in {\em undirected path graphs}, the class of vertex intersection graphs of a collection of paths in a tree), bipartite graphs, split graphs \cite{Bertossi84} and  $k$-trees with arbitrary $k$ \cite{CorneilSKeil87}, and it is {\it polynomial} in $k$-trees with fixed $k$, convex bipartite graphs \cite{DamaschkeMK90}, strongly chordal graphs~\cite{Farber84}.
Concerning the approximability, an \textbf{APX}-hardness proof with explicit inapproximability bound 
and a combinatorial 0.6-approximation algorithm is proposed in \cite{NguyenSHSMZ08}.
Better algorithms with approximation ratio 0.71 and 0.803 are given respectively in \cite{ChenENRRS13} and \cite{DBLP:conf/mfcs/AthanassopoulosCKK09}.
For any $\varepsilon > 0$, \textsc{Upper Edge Cover} is hard to approximate within a factor of $\frac{259}{260}+\varepsilon$ unless \textbf{P}=\textbf{NP} \cite{NguyenSHSMZ08}.
It admits a \textbf{PTAS} in $k$-trees (with arbitrary $k$), although \textsc{Upper Edge Cover} remains \textbf{APX}-complete
on $c$-dense graphs \cite{HeL13} (a graph is called {\em $c$-dense} if it contains at least $c\frac{n^2}{2}$ edges). 

\noindent
In contrast, for edge weighted graphs with non-negative weights, no result for
\textsc{Weighed Upper Edge Cover} is known, although some results are given
for \textsc{Maximum Weighted Spanning Star forest problem}: a 0.5-approximation is given in \cite{NguyenSHSMZ08}
(which is the best ratio obtained so far) and polynomial-time algorithms for special classes of
graphs such as trees and cactus graphs are presented in \cite{NguyenSHSMZ08,Nguyen15}.
Negative approximation results are presented in \cite{NguyenSHSMZ08,ChakrabartyG10,ChenENRRS13}.
In particular, \textsc{MaxWSSF} is \textbf{NP}-hard to approximate within $\frac{10}{11}+\varepsilon$~\cite{ChakrabartyG10}. Two generalizations of \textsc{WSSF},  denoted \textsc{MinExtWSSF} and \textsc{MaxExtWSSF}, have been introduced very recently in \cite{KhoshkhahGMT17} where the goal consists in {\em extending} some partial stars into spanning star forests.  In this context, a partial feasible solution is given in advance
and the goal is to extend this partial solution. Formally, the problem is defined as follow:

\begin{center}
\fbox{\begin{minipage}{.98\textwidth}
\noindent{\textsc{Extended weighted spanning star forest problem} (\textsc{ExtWSSF} in short)}\\\nopagebreak
{\bf Input:} A weighted graph $(G,w)$  and a packing of stars $\mathcal{U}=\{U_1,\dots,U_r\}$ where $G=(V,E)$ and $w(e)\geq 0$ for $e\in E$. \\\nopagebreak
{\bf Solution:} Spanning star forest  $\mathcal{S}=\{S_1,\dots,S_p\}\subseteq E$ containing $\mathcal{U}$. \\\nopagebreak
{\bf Output:}  $w(\mathcal{S})=\sum_{e\in \mathcal{S}}w(e)=\sum_{i=1}^p\sum_{e\in S_i}w(e)$.
\end{minipage}}
\end{center}

In  \cite{KhoshkhahGMT17}, several results have been given for both  {\em minimization}  (\textsc{MinExtWSSF}) and {\em maximization} (\textsc{MaxExtWSSF}) versions of \textsc{ExtWSSF} (denoted \textsc{MinExtWSSF} and \textsc{MaxExtWSSF} respectively). Dealing with  the minimization version for complete graphs: a dichotomy result of the computational complexity is presented depending on parameter $c$ of the (extended) $c$-relaxed triangle inequality and an FPT algorithm is  given.  For the maximization version, a positive approximation of  $1/2$ and a negative approximation result of $\frac{7}{8}$ (even for binary weights) are proposed.

A subset of vertices $V'$ is called \textit{non-blocking} if every vertex in $V'$ has at least one neighbor in $V\setminus V'$.
Actually, \textit{non-blocking} is dual of dominating set and vice versa.
For a given graph $G=(V,E)$ and a positive integer $k$, the \textsc{Non-blocker} problem asks if there is a \textit{non-blocking} set $V'\subseteq V$ with $|V'|\geq k$.
Hence, for unweighted graphs, optimal value of \textit{non-blocking} number equals the upper edge cover number.
In \cite{DehneFFPR06} Dehne et al. propose a parameterized perspective of the \textsc{Non-blocker} problem.
They give a linear kernel and an \textbf{FPT} algorithm running in  time $\mathcal{O^{*}}(2.5154^{k})$.
They also give faster algorithms for planar and bipartite graphs.


\textbf{\textit{Contributions:}}
The paper is organized in the following way.
We first show in Section~\ref{sec: complete graphs} that \textsc{Weighted Upper Edge Cover} in complete graphs is equivalent for its approximation to \textsc{MaxWSSF} in general graphs.
Then, we study \textsc{Weighted Upper Edge Cover} for bipartite graphs, split graphs and $k$-trees respectively in Sections~\ref{sec: bipartite graphs},~\ref{sec: split graphs} and~\ref{sec: k-trees}.
Motivated by the above results mostly negative, we propose a constant approximation ratio algorithm in Section~\ref{sec:apx-MaxDegreeDelta} for \textsc{Weighted Upper Edge Cover} in bounded degree graphs.


\noindent Note that all results given in this paper are valid  if $G$ is isolated vertex free instead of connected. 

\section{Complete graphs}\label{sec: complete graphs}

In this section, we deal with edge-weighted complete graphs.
This case seems to be the simplest one because the equivalence between \textsc{Upper Edge Cover} and \textsc{MaxSSF} for the unweighted case proven in~\cite{manlove1999algorithmic} remains valid for the weighted case as proven in the following.

\begin{theorem}\label{equivP_complete}
\textsc{MaxWSSF} in general graphs is equivalent to approximate \textsc{Weighted Upper Edge Cover} in complete graphs.
\end{theorem}
\begin{proof}
We propose two approximation preserving reductions, one from \textsc{MaxWSSF} in general graphs to \textsc{Weighted Upper Edge Cover} in complete graphs and the other from \textsc{Weighted Upper Edge Cover} to \textsc{MaxWSSF} in complete graphs.

\noindent $\bullet$ Reduction from \textsc{MaxWSSF} to \textsc{Weighted Upper Edge Cover} in complete graphs.

Let $(G,w)$ be an instance of \textsc{MaxWSSF} where $G=(V,E)$ is a connected graph with $n$ vertices, edge-weighted using $w$. We build an instance $(K_n,w')$ of \textsc{Weighted Upper Edge Cover} where $K_n$ is an edge-weighted complete graph $(V,E(K_n))$ over $n$ vertices, edge-weighted with $w'$, such that $\forall u,v\in V$ with $u\neq v$, $w'(uv)=w(uv)$ if $uv\in E$ and $w'(uv)=0$ otherwise.\\

Let $S^\prime\subseteq E(K_n)$ be a minimal edge cover of \textsc{Weighted Upper Edge Cover} with weight $w'(S^\prime)$.
The restriction of $S^\prime$ to $G$ gives a star spanning forest (with eventually trivial stars) $\mathcal{S}$.
Obviously, by construction we have:

\begin{equation}\label{eq1:theoequivP_complete}
w(\mathcal{S})=w'(S^\prime)
\end{equation}

Thus, from equality (\ref{eq1:theoequivP_complete}) we deduce $opt_{MaxWSSF}(G,w)\geq \uec(K_n,w')$.
Conversely, let $\mathcal{S}^*$ be an optimal  star spanning forest of \textsc{MaxWSSF} with value $opt_{MaxWSSF}(G,w)$.
By adding some edges from the center of some stars to the isolated vertices of $\mathcal{S}^*$, we obtain a minimal edge cover of $K_n$ of at least same value.
Hence,  $ \uec(K_n,w')\geq opt_{MaxWSSF}(G,w)$.
We can deduce,

\begin{equation}\label{eq2:theoequivP_complete}
 \uec(K_n,w')= opt_{MaxWSSF}(G,w)
\end{equation}

From equalities (\ref{eq1:theoequivP_complete}) and (\ref{eq2:theoequivP_complete}), we can deduce that any $\rho$ approximation of \textsc{Weighted Upper Edge Cover} for $(K_n,w')$ can be polynomially converted into a $\rho$ approximation of \textsc{MaxWSSF} for $(G,w)$.

\bigskip

\noindent $\bullet$ Reduction from  \textsc{Weighted Upper Edge Cover}  to  \textsc{MaxWSSF} in complete graphs.

From an edge-weighted complete graph $(K_n,w)$ instance of \textsc{Weighted Upper Edge Cover}, we set $(G,w')=(K_n,w)$ as an instance of  \textsc{MaxWSSF}. Since the graph is complete, the weights are non-negative  and the goal is maximization, we can only consider star spanning forests without trivial stars, i.e. minimal edge covers.
Hence, \textsc{Weighted Upper Edge Cover} is as a subproblem of \textsc{MaxWSSF}, even from an approximation point of view.
\end{proof}
\qed

From Theorem~\ref{equivP_complete} and from known results on \textsc{MaxWSSF} given in~\cite{NguyenSHSMZ08,ChakrabartyG10}, we deduce the following:

\begin{corollary}\label{cor:complete}
In complete graphs, \textsc{Weighted Upper Edge Cover} is  $1/2$-approximable but not approximable within $\frac{10}{11}+\varepsilon$ unless \textbf{P}$=$\textbf{NP}.
 \end{corollary}

\vspace{-0.4 cm}
\section{Bipartite graphs}\label{sec: bipartite graphs}
Let us now focus on bipartite graphs.
We prove that, even in bipartite graphs with binary weights, \textsc{Weighted Upper Edge Cover} is not $O(n^{\frac{1}{2}-\varepsilon})$ approximable unless \textbf{P} = \textbf{NP}.
Also, we show the problem is \textbf{APX}-complete even for bipartite graphs with fixed maximum degree $\Delta$.

\begin{theorem}\label{bip:Reduc_IS}
\textsc{Weighted Upper Edge Cover} in bipartite graphs with binary weights and cycle inequality is as hard \footnote{The reduction is actually a Strict-reduction and it is a particular A-reduction which preserves constant approximation.} as \textsc{MaxIS} in general graphs.
\end{theorem}
\begin{proof}
We propose an approximation preserving \textbf{APX}-reduction from \textsc{Independent Set} (denoted \textsc{MaxIS}) to \textsc{Weighted Upper Edge Cover}.

\noindent Given a connected graph $G=(V,E)$ with $n$ vertices and $m$ edges where $V=\{v_1,\dots,v_n\}$, instance of \textsc{MaxIS}, we build a connected  bipartite edge-weighted graph $H=(V_H,E_H,w)$ as follows (see also Figure \ref{Fig:BipTheoReduc_IS}):
\begin{itemize}
\item For each $v_i\in V$, add a $P_3$ with edge set $\{v_iv_{i,1},v_{i,1}v_{i,2}\}$.
\item For each edge $e=v_iv_j\in E$ where $i<j$, add a middle vertex $v_{ij}$ on edge $e$.
\item $w(e): =\begin{cases}1 & \text{if }  e=v_iv_{i,1}$ for some $v_i\in V\\0 & \text{otherwise}.\end{cases}$
\end{itemize}

Clearly, $H$ is a connected bipartite graph on $|V_H|=3n+m$ vertices and $|E_H|=2(m+n)$ edges. Moreover,  weights are binary and instance satisfies cycle inequality.

\vspace{-0.5 cm}
\begin{figure}\label{fig:reduction_IS}
\centering
\begin{tikzpicture}[scale=0.9, transform shape]
\node[vertex1] (v1) at (0,0) {$v_1$};
\node[vertex1, right of=v1, node distance=2 cm] (v2){$v_2$};
\node[vertex1, below of=v2, node distance=2 cm] (v3){$v_3$};
\node[vertex1, below of=v1, node distance=2 cm] (v4){$v_4$};
\draw (v1)--(v2);
\draw (v2)--(v3)--(v4)--(v1)--(v3);

\node() at (1,-2.75) {$G=(V,E)$};
\node[vertex1] (v1H) at (5.5,0.2) {$v_1$};
\node[vertex, right of=v1H, node distance=1.2 cm] (v12H){$v_{12}$};
\node[vertex1, right of=v12H, node distance=1.2 cm] (v2H){$v_2$};
\node[vertex, below of=v2H, node distance=1.2 cm] (v23H){$v_{23}$};
\node[vertex1, below of=v23H, node distance=1.2 cm] (v3H){$v_3$};
\node[vertex, below of=v1H, node distance=1.2 cm] (v14H){$v_{14}$};
\node[vertex1, below of=v14H, node distance=1.2 cm] (v4H){$v_4$};
\node[vertex, right of=v4H, node distance=1.2 cm] (v34H){$v_{34}$};
\node[vertex, right of=v14H, node distance=1.2 cm] (v13H){$v_{13}$};
\node[vertex, left of=v1H, node distance=1 cm] (v1x){$v_{1,1}$};
\node[vertex, left of=v1x, node distance=1 cm] (v1y){$v_{1,2}$};
\node[vertex, right of=v2H, node distance=1 cm] (v2x){$v_{2,1}$};
\node[vertex, right of=v2x, node distance=1 cm] (v2y){$v_{2,2}$};
\node[vertex, left of=v4H, node distance=1 cm] (v4x){$v_{4,1}$};
\node[vertex, left of=v4x, node distance=1 cm] (v4y){$v_{4,2}$};
\node[vertex, right of=v3H, node distance=1 cm] (v3x){$v_{3,1}$};
\node[vertex, right of=v3x, node distance=1 cm] (v3y){$v_{3,2}$};
\draw (v1H)--(v12H)--(v2H)--(v23H)--(v3H)--(v34H)--(v4H)--(v14H)--(v1H)--(v13H)--(v3H);
\draw (v1H)--(v1x)--(v1y);
\draw (v2H)--(v2x)--(v2y);
\draw (v3H)--(v3x)--(v3y);
\draw (v4H)--(v4x)--(v4y);
\draw (v1H) -- node[midway,above] {$0$} (v12H);
\draw (v2H) -- node[midway,above] {$0$} (v12H);
\draw (v2H) -- node[midway,right] {$0$} (v23H);
\draw (v3H) -- node[midway,right] {$0$} (v23H);
\draw (v3H) -- node[midway,below] {$0$} (v34H);
\draw (v4H) -- node[midway,below] {$0$} (v34H);
\draw (v1H) -- node[midway,left] {$0$} (v14H);
\draw (v4H) -- node[midway,left] {$0$} (v14H);
\draw (v1H) -- node[midway,right] {$0$} (v13H);
\draw (v3H) -- node[midway,right] {$0$} (v13H);
\draw (v1H) -- node[midway,above] {$1$} (v1x);
\draw (v1x) -- node[midway,above] {$0$} (v1y);
\draw (v2H) -- node[midway,above] {$1$} (v2x);
\draw (v2x) -- node[midway,above] {$0$} (v2y);
\draw (v3H) -- node[midway,below] {$1$} (v3x);
\draw (v3x) -- node[midway,below] {$0$} (v3y);
\draw (v4H) -- node[midway,below] {$1$} (v4x);
\draw (v4x) -- node[midway,below] {$0$} (v4y);

\node() at (6.75,-2.75) {$H=(V_H,E_H,w)$};
\end{tikzpicture}
\caption{Construction of $H$ from $G$. The weights are indicated on edges.}\label{Fig:BipTheoReduc_IS}
\end{figure}
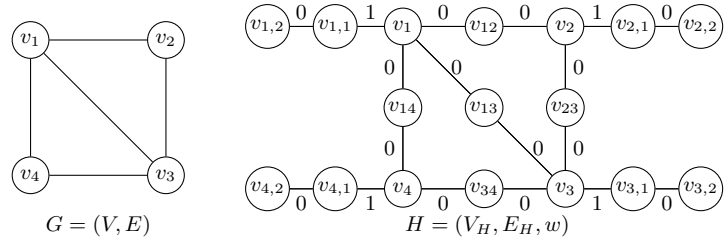
\noindent Let $S^*$ be a maximum independent set of $G$ with  size $\alpha(G)$. For each $e\in E$, let $v^e\in V\setminus S^*$ be a vertex which covers $e$; it is possible since $V\setminus S^*$ is a {\em vertex cover} of $G$. Moreover, $\{v^e:e\in E\}=V\setminus S^*$ since $S^*$ is a maximum independent set of $G$. Clearly, $S^\prime=\{v_{xy}v^e :e=xy\in E\}\cup \{v_{i,1}v_{i,2}:v_i\in V\}\cup \{v_iv_{i,1}:v_i\in S^*\}$ covers all vertices of $H$ and since it doesn't include any $P_3$, then $S^\prime$ is a minimal edge cover of $H$. By construction, $w(S^\prime)=|S^*|=\alpha(G)$. Hence, we deduce:
\begin{equation}\label{eq1:BipReduc_IS}
\uec(H,w)\geq \alpha(G)
\end{equation}

Conversely, suppose $S^\prime$ is a minimal edge cover  of $H$ with weight $w(S^\prime)$.
Let us make some simple observations of every minimal edge cover  of $H$.
Clearly, $\{v_{i1}v_{i2}:v_i\in V\}$ is part of every feasible solution because $v_{i2}$ for $v_i\in V$ are leaves of $H$.
Moreover, for each $e=v_iv_j\in E$ with $i<j$,  at least one edge between $v_iv_{ij}$ or $v_jv_{ij}$ belongs to any minimal edge cover of $H$. If $v_iv_{ij}\notin S'$, it implies that $v_jv_{j,1}\notin S'$ is not a part of the feasible solution because of minimality of $S'$. Hence, $S=\{v_i:v_iv_{i,1}\in S'\}$ is an independent set of $G$ with size $|S|=w(S^\prime)$. We deduce:

\begin{equation}\label{eq2:BipReduc_IS}
\alpha(G)\geq \uec(H,w)
\end{equation}

Using inequalities (\ref{eq1:BipReduc_IS}) and (\ref{eq2:BipReduc_IS}) we deduce:

\begin{equation}\label{eq3:BipReduc_IS}
\alpha(G)=\uec(H,w)
\end{equation}

In conclusion, for each minimal edge cover $S^\prime$  on $H$, there is an independent set $S$ of $G$ (computed in polynomial-time) such that $|S|\geq w(S^\prime)$.
\end{proof}

From Theorem \ref{bip:Reduc_IS}, we immediately deduce that \textsc{Weighted Upper Edge Cover} in bipartite graphs is not in \textbf{APX} unless \textbf{P}$=$\textbf{NP}.
However, using several results  \cite{GJ79,AlimontiK00} concerning the \textbf{APX}-completeness of \textsc{MaxIS} in connected graph $G$ with constant maximum degree $\Delta(G)\geq 3$ or \textbf{NP}-completeness of \textsc{MaxIS} in planar graphs, we obtain:

\begin{corollary}\label{Gen_NPc:BipReduc_IS}
\textsc{Weighted Upper Edge Cover} in bipartite (resp., planar  bipartite) graphs  of maximum degree $\Delta$ for any fixed $\Delta\geq 4$ and binary weights is \textbf{APX}-complete (resp. \textbf{NP}-complete).
\end{corollary}

\begin{proof}
Let us revisit the construction given in Theorem~\ref{bip:Reduc_IS}.
If the instance of \textsc{MaxIS} has maximum degree 3 (resp. is planar with maximum degree 3), then the constructed instance of \textsc{Weighted Upper Edge Cover} is a bipartite (resp., planar bipartite) graph of maximum degree 4.
\qed
\end{proof}

\noindent Using the strong inapproximation result for \textsc{MaxIS} given in \cite{Zuckerman07}, and because the reduction given in previous theorem is indeed a gap-reduction, we also deduce:

\begin{corollary}\label{Gen_inapproxBip:Reduc_IS}
For any $\varepsilon>0$, \textsc{Weighted Upper Edge Cover} in bipartite graphs of $n$ vertices is not $O(n^{\frac{1}{2}-\varepsilon})$ approximable unless \textbf{P} = \textbf{NP}, 
even for binary weights and cycle inequality.
\end{corollary}

\begin{proof}
We use the reduction given in Theorem \ref{bip:Reduc_IS} and the inapproximability of \textsc{MaxIS}.
\textsc{MaxIS} is known to be
hard to approximate~\cite{Zuckerman07}.
In particular, it is known that, for all $\varepsilon>0$, it is \textbf{NP}-hard to distinguish for an
$n$-vertex graph $G$ between $\alpha(G)>n^{1-\varepsilon}$ and
$\alpha(G)<n^{\varepsilon}$.

In the construction of $H$ (see Figure \ref{Fig:BipTheoReduc_IS}), we know that $|V_H|=m+3n$ and $|E_H|=2(m+n)$ where $m,n$ are numbers of the
edges and vertices of $G$ respectively.
Hence, we deduce $|V_H|\leq 2n^{2}$, and the claimed result follows.
\qed
\end{proof}

\noindent We also deduce one inapproximability result depending on the maximum degree.

\begin{corollary}\label{Maxdegree:BipReduc_IS}
For any constant $\varepsilon > 0$, unless
\textbf{NP}$\subseteq$\textbf{ZPTIME}$(n^{\mathrm{poly} \log n})$, it is hard to approximate \textsc{Weighted Upper Edge Cover} on bipartite graphs of maximum degree $\Delta$ within a factor of $\Theta\left(\frac{1}{\Delta^{1-\varepsilon}}\right)$.
\end{corollary}

\begin{proof}
We will prove that it is difficult for a graph $H$ (even bipartite with binary weights) of maximum degree $\Delta$ to  distinguish between the following two cases:
\begin{itemize}
\item[$\bullet$](\textrm{Yes-Instance})$\uec(H,w)\geq \frac{|V(H)|}{\Delta(G)^{1+\varepsilon}}$, \\
\item[$\bullet$](\textrm{No-Instance}) $\uec(H,w)\leq \frac{|V(H)|}{\Delta(G)^{2-\varepsilon}}$.
\end{itemize}

Hence, the result consists of showing that the transformation given in Theorem \ref{bip:Reduc_IS} is a gap reduction.
It is proved that: Let $\tau(n)$ be any function from integers to integers. Assuming that \textbf{NP}$\nsubseteq$\textbf{ZPTIME}$(n^{O(\tau(n))})$, there is no polynomial-time algorithm that can solve the following problem \cite{ChalermsookLN13} (Theorem 5.7, adapted from \cite{Trevisan01}). For any constant $\varepsilon > 0$ and any integer $q$, given a regular graph $G$ of size $q^{O({\tau(n)})}$ such that all vertices have degree $\Delta=2^{O(\tau(n))}$, the goal is to distinguish between the following two cases:

\begin{itemize}
\item[$\bullet$](\textrm{Yes-Instance}) $\alpha(G)\geq \frac{|V(G)|}{\Delta^{\varepsilon}}$, \\
\item[$\bullet$](\textrm{No-Instance}) $\alpha(G)\leq \frac{|V(G)|}{\Delta^{1-\varepsilon}}$.
\end{itemize}

\noindent Note that if $G$ is a $\Delta$-regular graph, then graph $H$ resulting of Theorem \ref{bip:Reduc_IS} is a bipartite graph of maximum degree $\Delta+1=\Theta(\Delta)$. Thus, since $\alpha(G) = \uec(H,w)$ and $|V(H)|=3|V(G)|+|E(G)|=\Theta(\Delta |V(G)|)$, we get the expected result.
\qed
\end{proof}
\vspace{-0.4 cm}
\section{Split graphs}\label{sec: split graphs}

We will now focus on split graphs.
Recall that a graph $G=(L\cup R,E)$ is a split graph if the subgraph induced by $L$ and $R$ is a maximum clique and an independent  set respectively.

\begin{theorem}\label{split:Reduc_IS}
\textsc{Weighted Upper Edge Cover}  in split graphs  with binary weights and cycle inequality is as hard \footnote{The reduction is actually a Strict-reduction and it is a particular A-reduction which preserves constant approximation.} as \textsc{MaxIS} in general graphs.
\end{theorem}

\begin{proof}
The proof is based on a reduction from \textsc{MaxIS}. Given a graph $G=(V,E)$ of $n$ vertices and $m$ edges where $V=\{v_1,\dots,v_n\}$
and $E=\{e_1,\dots,e_m\}$, instance of \textsc{MaxIS}, we build a split weighted graph $H=(V_H,E_H,w)$ as follows:
\begin{itemize}

\item Put two copies of vertices $V$ in $H$, indicated by $C=\{c_1,\dots,c_n\}$ and $C'=\{c'_1,\dots,c'_n\}$ and make two cliques of size $n$
such that all pairs of vertices in $C$ and $C'$ are connected to each other with edges of weight $0$.

\item Connect all pairs $c_ic'_j$ for $1\leq i,j\leq n$ with edges of weight $1$ to make a clique of size $2n$.

\item Add a set of $m$ new vertices $\{p_1,\dots,p_m\}$ corresponding to edges of $E$ and connect $p_i$ to $c_j,c_k$ with edges of weight $0$ if $e_i=v_jv_k\in E$.

\item Add a set of $n$ new vertices $\{t_1,\dots,t_n\}$ and connect each $t_i$ to $c'_i$ with edges of weight $0$.
\end{itemize}

It is easy to check $H$ is a weighted split graph with binary weights and cycle inequality which contains a clique of size $2n$ and an
independent set of size $n+m$. Figure \ref{Fig:split_Reduc_IS} gives an illustration of the construction of $H$ from a $P_3$.

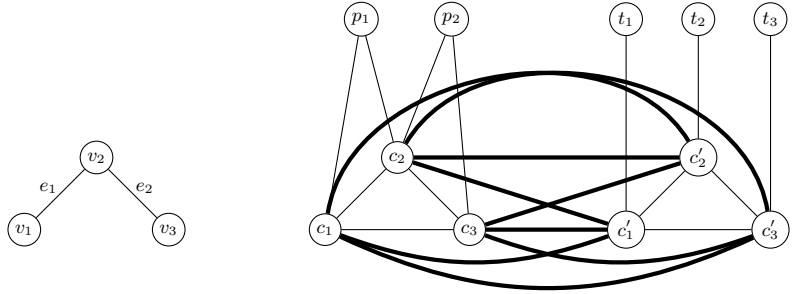
\begin{figure}
\centering
\begin{tikzpicture}[scale=0.8, transform shape]
\tikzstyle{vertex}=[circle, draw, inner sep=2pt,  minimum width=1 pt, minimum size=0.1cm]
\tikzstyle{vertex1}=[circle, draw, inner sep=2pt, fill=black!100, minimum width=1pt, minimum size=0.1cm]

\node[vertex] (v1) at (-1,0) {$v_1$};
\node[vertex] (v2) at (0.2,1.2) {$v_2$};
\node[vertex] (v3) at (1.4,0) {$v_3$};

\draw (v1)--(v2);
\draw (v2)--(v3);
\node () at (-0.6,0.7) {$e_1$};
\node () at (1,0.7) {$e_2$};

\node[vertex] (c1) at (4,0) {$c_1$};
\node[vertex] (c2) at (5.2,1.2) {$c_2$};
\node[vertex] (c3) at (6.4,0) {$c_3$};

\node[vertex] (cp1) at (9,0) {$c'_1$};
\node[vertex] (cp2) at (10.2,1.2) {$c'_2$};
\node[vertex] (cp3) at (11.4,0) {$c'_3$};

\draw (c1)--(c2)--(c3)--(c1);
\draw (cp1)--(cp2)--(cp3)--(cp1);
\draw (c2)[ultra thick]--(cp2);
\draw (c3)[ultra thick]--(cp1);
\draw (c2)[ultra thick]--(cp1);
\draw (c3)[ultra thick]--(cp2);
\draw (c1) edge[ultra thick, bend right=20] (cp1);
\draw (c1) edge[ultra thick, bend right=25] (cp3);
\draw (c3) edge[ultra thick, bend right=20] (cp3);
\draw (c1) edge[ultra thick, bend left=70] (cp2);
\draw (c2) edge[ultra thick, bend left=70] (cp3);

\node[vertex] (p1) at (4.6,3.5) {$p_1$};
\node[vertex] (p2) at (6.1,3.5) {$p_2$};
\draw (c1)--(p1)--(c2)--(p2)--(c3);

\node[vertex] (t1) at (9,3.5) {$t_1$};
\node[vertex] (t2) at (10.2,3.5) {$t_2$};
\node[vertex] (t3) at (11.4,3.5) {$t_3$};

\draw (cp1)--(t1);
\draw (cp2)--(t2);
\draw (cp3)--(t3);
\end{tikzpicture}
\caption{Construction of split graph $H=(V_H,E_H)$ from a $P_3$.
The weights of thick edges in $H$ are 1 and for the others are 0.}\label{Fig:split_Reduc_IS}
\end{figure}

\bigskip
We claim that $G$ has an independent set of size $k$ iff there exists a minimal edge cover of $H$ with total weight $k$.
\bigskip

Let $S$ be an independent set of $G$ with size $|S|$. For each $e_i\in E$, there is a vertex $v_{e_i}\notin S$ which covers $e_i$
since $S$ is an independent set of $G$. Consider the set $\{c_{e_i}:v_{e_i}\notin S \}$ of vertices in $C$ corresponding to vertices of $V\setminus S$,
$S^\prime=\{c_{e_i}p_i:e_i\in E\}\cup \{c'_it_i:v_i\in V\}\cup \{c_ic'_i:v_i\in S\}$ is a minimal edge cover of $H$.
By construction, $w(S^\prime)=|S|$. Hence, we deduce:

\begin{equation}\label{eq1:SplitReduc_IS}
\uec(H,w)\geq \alpha(G)
\end{equation}

Conversely, let  be a minimal edge cover of $H$ with weight $w(S^\prime)$.
Since for $1\leq i\leq n$, $t_i$'s  are leaves in $H$, $\{t_ic'_i:v_i\in V\}$ is a part of $S^\prime$.
Moreover, for each $e_k=v_iv_j\in E$ with $i<j$,  at least one edge among $c_ip_k$ or $c_jp_k$
belongs to $S^\prime$. W.l.o.g., assume that $c_ip_k\in S^\prime$; this means that $c_ic'_j\notin S'$ for all $1\leq j\leq n$.
Furthermore, for each $c_i\in C$ at most one edge $c_ic'_j\in S^\prime$ for $1\leq j\leq n$.
Hence, $S=\{v_i:c_ic'_j\in S'\}$ is an independent set of $G$ with size $|S|=w(S^\prime)$. We deduce,

\begin{equation}\label{eq2:SplitReduc_IS}
\alpha(G)\geq \uec(H,w)
\end{equation}

Using inequalities (\ref{eq1:SplitReduc_IS}) and (\ref{eq2:SplitReduc_IS}) we deduce $\alpha(G)=\uec(H,w)$.
\qed
\end{proof}


\begin{corollary}\label{Gen_inapproxSPLIT:Reduc_IS}
 \textsc{Weighted Upper Edge Cover} in split 3-subregular graphs  is \textbf{APX}-complete and for any $\varepsilon>0$, \textsc{weighted upper edge cover}
in split graphs of $n$ vertices is not $O(n^{\frac{1}{2}-\varepsilon})$ approximable unless \textbf{P} = \textbf{NP}.
\end{corollary}

\vspace{-0.4 cm}
\section{$k$-trees}\label{sec: k-trees}

Recall that a $k$-tree is a graph which results from the following inductive definition: A $K_{k+1}$ is a $k$-tree.
If a graph $G$ is a $k$-tree, then the addition of a new vertex which has exactly $k$ neighbors in $G$ such that these $k+1$ vertices induce a $K_{k+1}$ forms a $k$-tree.
As a main result in this section we prove \textsc{Weighted Upper Edge Cover} is \textbf{APX}-complete in $k$-trees even for trivalued weights.

\vspace{-0.4 cm}
\subsection{Negative approximation result}

From Corollary \ref{cor:complete},  we already know that  \textsc{Weighted Upper Edge Cover} is \textbf{NP}-hard to approximate within a ratio strictly better than $\frac{10}{11}$ because the class of all $k$-trees contains the class of complete graphs.
However, this lower bound needs a non-constant number of distinct values \cite{ChakrabartyG10}. Here, we strengthen the result by proving the
existence of lower bounds even for 3 distinct weights. On the other hand, \textsc{Weighted Upper Edge Cover} in weighted {complete graphs} and {$k$-trees} with binary weights is not strictly  approximable within ratio better than $\frac{259}{260}\approx 0.9961$ because it is proved
in \cite[Theorem 3.6]{NguyenSHSMZ08} a lower bound of $\frac{259}{260}+\varepsilon$ for \textsc{MaxSSF}.
Here, we slightly improve this latter bound to  $\frac{179}{190}\approx 0.9421$  of \textsc{Weighted Upper Edge Cover}  with trivalued weights for $k$-trees.

\begin{theorem}\label{Hardktree:theo}
\textsc{Weighted Upper Edge Cover}  is \textbf{APX}-hard in the class of $k$-trees, even for trivalued weights.
\end{theorem}

\begin{proof}
We give an approximation preserving reduction from independent set problem. It is known that \textsc{MaxIS} is \textbf{APX}-complete in graphs of maximum degree 3 \cite{AlimontiK00}.

Let $G=(V,E)$ be an instance of \textsc{MaxIS} where $G$ is a connected graph of maximum degree $3$ of $n\geq 3$ vertices and $m$ edges.
We build a weighted graph $G'=(V',E',w)$ for \textsc{Weighted Upper Edge Cover} problem where $V'=V'_c\cup V'_E$ and $E'=E'_c\cup (\cup_{e\in E}E'_e)$ as follows:
\begin{itemize}
\item $V'_c=\{v':v\in V\}$ and $V'_E=\cup_{e\in E} V'_e$ where $V'_e=\{v_{e,1},\ldots,v_{e,(n-1)}\}$.
\item The subgraph  $G'[V'_c]=(V'_c,E'_c)$  induced by $V'_c$ is a $K_n$.
\item For each $e=uv\in E$, let us describe the edge set $E'_e$:
\begin{itemize}
\item for every $i=1,\dots,n-1$, vertex $v_{e,i}$ is linked to $u'$ and $v'$.
\item vertex $v_{e,1}$ is linked to the subset $S_{e,1}= V'_c \setminus \{u',v'\}$.
\item for every $i=2,\dots,n-1$, vertex $v_{e,i}$ is linked to $\{v_{e,1},\ldots,v_{e,i-1}\}$ and an arbitrary subset $S_{e,i}\subset S_{e,(i-1)}$ of size $n-i-1$.
\end{itemize}
\end{itemize}

The weight $w(xy)$ for $xy\in E'$ is given by:

$w(xy)=\begin{cases}n-1 & xy\in E'_c,\\1&  xy\in E'_e \text{ with } e=uv\in E \text{ and } x\in\{u',v'\},y\in V'_e,\\ 0 & \text{otherwise}.\end{cases}$

\bigskip

\noindent Note that $|V'|=n+m(n-1)$ and clearly $G'$ can be constructed from $G$ in polynomial time. $G'$ is a $n$-tree because initially all $V'_c\cup \{v_{e,1}\}$ are clique of size $n+1$ for $e\in E$ and at each step the addition of $v_{e,i+1}$ maintains a $K_{n+1}$ containing $v_{e,i+1}$
in the subgraph induced by $V'_c\cup \{v_{e,j}:e\in E,j\leq i\}$. Figure \ref{Fig:k-tree} proposes an illustration of this construction for a $P_3$.

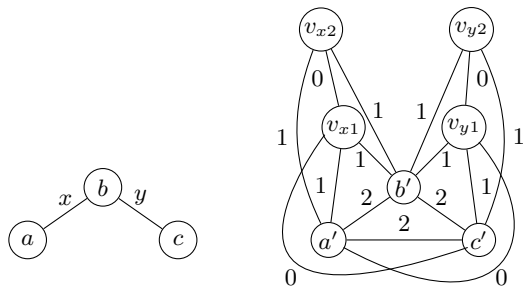
\begin{figure}
\centering
\begin{tikzpicture}[]
\node[vertex1] (a) at (0,0) {$a$};
\node[vertex1] (b) at (1,0.7) {$b$};
\node[vertex1] (c) at (2,0) {$c$};

\draw (a) -- node[midway,above] {$x$} (b);
\draw (b) -- node[midway,above] {$y$} (c);

\node[vertex] (ap) at (4,0) {$a'$};
\node[vertex] (bp) at (5,0.7) {$b'$};
\node[vertex] (cp) at (6,0) {$c'$};

\node[vertex] (v11) at (4.2,1.5) {$v_{x1}$};
\node[vertex] (v12) at (3.9,2.8) {$v_{x2}$};

\node[vertex] (v21) at (5.8,1.5) {$v_{y1}$};
\node[vertex] (v22) at (5.9,2.8) {$v_{y2}$};

\draw (ap) -- node[midway,above] {$2$} (bp);
\draw (bp) -- node[midway,above] {$2$} (cp);
\draw (cp) -- node[midway,above] {$2$} (ap);
\draw (ap) -- node[midway,left] {$1$} (v11);
\draw (bp) -- node[midway,left] {$1$} (v11);
\draw (bp) -- node[midway,right] {$1$} (v21);
\draw (cp) -- node[midway,right] {$1$} (v21);
\draw (v11) -- node[midway,left] {$0$} (v12);
\draw (v21) -- node[midway,right] {$0$} (v22);
\draw (v12) edge[bend right=23] node[midway,left]{$1$}(ap);
\draw (v12) -- node[midway,right] {$1$} (bp);
\draw (v22) edge[bend left=24] node[midway,right]{$1$}(cp);
\draw (v22) -- node[midway,left] {$1$} (bp);
\draw (3.95,1.4) .. controls (3.5,0.8) and (2.3,-1.3) .. (5.85,-0.1);
\node () at (3.5,-0.5) {$0$};
\draw (6,1.3) .. controls (7,0.1) and (6.4,-1.3) .. (4.2,-0.1);
\node () at (6.3,-0.5) {$0$};
\end{tikzpicture}
\caption{The constructed weighted graph $G'=(V',E'.w)$ (right side) build from a $P_3$  $G=(V=\{a,b,c\},E=\{1,2\})$
(left side).}\label{Fig:k-tree}
\end{figure}
\bigskip

\noindent We are going to prove that any $\rho$-approximation for \textsc{Weighted Upper Edge Cover} in $k$-Trees can be polynomially converted into a $(\frac{11}{2} \rho -\frac{9}{2})$-approximation for \textsc{MaxIS} in graphs of maximum degree 3.\\

First, consider an arbitrary independent set $S$ of $G$. From $S$ we build a minimal edge cover $F$ of $G'$ of size at least $(n-1)(|S|+m)$. For each $e=uv\in E$, there is a vertex $f(e)\in ((V\setminus S)\cap \{u,v\})$ because $S$ is an independent set; choose arbitrarily such vertex $r\in V\setminus S$. We set $F=\{f(e)'v_{f(e),i}:e\in E,i\leq n-1\}\cup \{r'v':v\in (V\setminus X)\}$ where $X=\{f(e):e\in E\}\cup \{r\}$.  We deduce $\uec(G',w')\geq w(F)=(n-1)m+(n-1)|S|$ and considering $S$ as a maximum independent set induces:

\begin{equation}\label{Hardktree:theoeq1}
\uec(G',w')\geq (n-1)(m+\alpha(G))
\end{equation}

\bigskip

Conversely, assume that $F$ is a minimal edge cover of $G'$. We will polynomially modify $F$ into another minimal edge cover  $F'$ of better weight.
\begin{property}\label{Hardktree:theoproperty1}
We can assume that $F$ satisfies the following  facts:
\begin{itemize}
\item[$(a)$] for each $e=uv\in E$ at least one of $u'$ or $v'$ is center,
\item[$(b)$] for each $e=uv\in E$, any vertex of $V'_e$ is a leaf and its center is $u'$ or $v'$.
\end{itemize}
\end{property}

\begin{proof}
For $(a)$ Otherwise, we could modify $F$ into $F'$ by repeating the following process for each edge $uv\in E$ where $u'$ and $v'$ are leaves in $F$ to satisfy $(a)$: if none of centers of $u$ and $v$ are in $V'_e$, then $t=u$ else $t$ is one of $u$, $v$ which its center is in $V'_e$. Let $S=\{ab\in F: a\in V'_e\cup \{t\}\}$ and $S'=\{tx: x\in V'_e\}$. Now $F'=(F\setminus S)\cup S'$ remains a  minimal edge cover of $G'$ that $w(F')\geq w(F)$ and $t$ is a center in $F'$. \\

For $(b)$ Let $e=uv\in E$ and w.l.o.g $u$ is a center in $F$. Let $S=\{ab:a\in V'_e\}$ and $S'=\{ux:x\in V'_e\}$. Now $F'=(F\setminus S')\cup S$ is a spanning star forest with possibly trivial stars of $G'$ with $w(F')\geq w(F)$ which satisfies $(b)$. Notice after these stages, we may create of some isolated vertices included in $V'_c$. However, connecting every isolated vertices in $V'_c$ to an arbitrary center in $V'_c$ induces a minimal edge cover with larger weight.
\qed
\end{proof}

Let $I'\subseteq V'_c$ be the leaves of the stars of $F'$. By considering $(a)$ in Property \ref{Hardktree:theoproperty1}, $I=\{v:v'\in I'\}$ is an independent set of $G$. Since for each minimal edge cover  $F$, there exist a minimal edge cover  $F^\prime$ such that:

\begin{equation}\label{Hardktree:theoeq2}
w(F)\leq w(F^\prime)= (m+|I|)(n-1)\leq (m+\alpha(G))(n-1)
\end{equation}
 Hence by considering inequality (\ref{Hardktree:theoeq1}) $\uec(G',w')= (m+\alpha(G))(n-1)$.

Let $F$ be a $\rho$-approximation for \textsc{Weighted Upper Edge Cover} for $(G',w')$ and $I$ be an independent set of $G$ which made by $F'$ then:

\begin{equation} \label{Hardktree:theoeq3}
\rho\leq \frac{w(F)}{\uec(G',w')}\leq \frac{w(F')}{\uec(G',w')}=\frac{(n-1)(m+|I|)}{(n-1)(m+\alpha(G))}=\frac{m+|I|}{m+\alpha(G)}
\end{equation}

since $G$ is connected of maximum degree 3, we know $n\leq 3 \alpha(G)$ (using Brook's Theorem and $n\geq 5$),
and then $m\leq  \frac{9}{2}  \alpha(G)$. Using this:

\begin{equation*}
\Rightarrow 1-\rho \geq \frac{\alpha(G)-|I|}{m+\alpha(G)}\geq \frac{\alpha(G)-|I|}{11/2\alpha(G)}
\end{equation*}

\begin{equation*}
\Rightarrow \frac{11}{2} \rho -\frac{9}{2}  \leq \frac{|I|}{\alpha(G)}
\end{equation*}

or equivalently $\frac{|I|}{\alpha(G)}\geq \frac{11}{2} \cdot\frac{w(F)}{\uec(G',w')}-\frac{9}{2}$.
Hence, \textsc{Weighted Upper Edge Cover}  is \textbf{APX}-hard in the class of  weighted $k$-Trees with trivalued weights.
\end{proof}
\qed

\begin{corollary}\label{coro:neqktree}
\textsc{Weighted Upper Edge Cover}  is not approximable within $\frac{179}{190}+\varepsilon$  for every $\varepsilon>0$ unless \textbf{P}$=$\textbf{NP} in the class of  weighted $k$-trees, even if there are only three distinct weights.
\end{corollary}

\begin{proof}
In \cite{ChlebikC06} it is proved  \textsc{MaxIS} is not $\frac{94}{95}+\varepsilon$ in graphs of maximum degree 3,
even in cubic connected graphs. Using $\rho'=\frac{94}{95}$ and  $\rho'\geq \frac{11}{2} \rho -\frac{9}{2}$ given in Theorem \ref{Hardktree:theo} produces a lower bound $\rho=\frac{179}{190}$.
\qed
\end{proof}

\vspace{-0.4 cm}
\subsection{Positive approximation result}

Now, we propose positive approximation result of \textsc{Weighted Upper Edge Cover} via the use of an approximation preserving reduction
from \textsc{MaxWSSF} which polynomially transform any $\rho$-approximation into a
$\frac{k-1}{k+1}\rho$-approximation for  \textsc{weighted upper edge cover}.

\begin{theorem}\label{theo:apx-k-trees}
In $k$-trees, \textsc{Weighted Upper Edge Cover} is $\frac{k-1}{2(k+1)}$-approximable.
\end{theorem}

\begin{proof}
The proof uses an approximation preserving reduction
from \textsc{MaxWSSF} which polynomially transform any $\rho$-approximation into a
$\frac{k-1}{k+1}\rho$-approximation for  \textsc{weighted upper edge cover}. Then,
using the $0.5$-approximation of \textsc{MaxWSSF} given in  \cite{NguyenSHSMZ08}, we will get the expected result. \\

Consider an edge-weighted $k$-tree $(G,w)$ where $G=(V,E)$ and assume $G$ is not complete. Let  $\mathcal{S}=\{S_1,\dots,S_r\}\subseteq E$ be a nice spanning star forest of $(G,w)$ (see Property \ref{niceSSF}) which is a $\rho$-approximation of \textsc{MaxWSSF}, that is:

\begin{equation}\label{eq1:approxk-tree}
w(\mathcal{S})\geq \rho\cdot opt_{MaxWSSF}(G,w)
\end{equation}

Now, we show how to modify $\mathcal{S}$ into a minimal edge cover $S$ without loosing too much.\\

Before, we need to introduce some definitions and notations. A {\em vertex-coloring} $\mathcal{C}=(C_1,\dots,C_q)$ of a graph $G$ is a partition of vertices into independent sets (called {\em colors}). The {\em chromatic number} of $G$, denoted $\chi(G)$, is the minimum number of colors used in a vertex-coloring. If $G$ is a $k$-tree, it is well known that $\chi(G)=k+1$ and such an optimal vertex-coloring can be done in linear time; hence, consider any optimal vertex-coloring $\mathcal{C}=\{C_1,\dots,C_{k+1}\}$ of $G$. Moreover, in $k$-trees
we know that each vertex $u\in C_i$ of color $i$ is adjacent to some vertex $v\in C_j$ of color $j$ for every $j\neq i$. We color the edges of $E(\mathcal{S})$ incident to every isolated vertices of  $\mathrm{Triv}$ using the $k+1$ colors where the color of such edge is given by the same color of its leaf. Formally, let $E^\prime=\{uv\in E\colon v\in \mathrm{Triv}\}\subseteq  E(\mathcal{S})$ be the subset of edges incident to isolated vertices $\mathrm{Triv}$ and let $E_i=\{cv=e_v(\mathcal{S})\in E(\mathcal{S})\colon v\in C_i\setminus \mathrm{Triv}\}$ for every $i\leq k+1$ where $c$ is some center of $\mathcal{S}$. The key property is the following:

\begin{property}\label{property1:k-tree}
for any $i<i'$, by deleting some edges of $E_i\cup E_{i'}$ and by adding edges from  $E^\prime$ we obtain a minimal edge cover.
\end{property}

\begin{proof}
It is valid because each vertex of color $i$ is adjacent to some vertices of every other colors. Formally, fix two indices $1\leq i<i'\leq k+1$. Iteratively apply the following procedure:  consider $v\in \mathrm{Triv}$; there is $u\in V\setminus  \mathrm{Triv}$ such that $u\in C_i\cup C_{i'}$ (say $C_i$) and $vu\in E$. By hypothesis, $u$ is a leaf of some $\ell$-star $S_r$ of $\mathcal{S}$. If at this stage $\ell\geq 2$, then add edge $uv\in E^\prime$ and delete edge $uc\in E_i$ of  color $i$; otherwise $\ell=1$ and we just add edge $uv\in E^\prime$. At the end, we get a minimal edge cover.
\qed
\end{proof}

Now, consider $i_1,i_2$ with $i_1<i_2$ such that $w(E_{i_1}\cup E_{i_2})=\min\{w(E_i\cup E_{i'})\colon 1\leq i<i'\leq k+1\}$. Using Property \ref{property1:k-tree}
we can polynomially find a minimal edge cover $S$ of $(G,w)$. By construction, $\sum_{i=1}^{k+1} w(E_i)\leq w(E(\mathcal{S}))$ and then:

\begin{equation}\label{eq2:approxk-tree}
w(E_{i_1}\cup E_{i_2})\leq \frac{2}{k+1}w(E(\mathcal{S}))
\end{equation}

Hence using inequalities (\ref{eq1:approxk-tree}) and (\ref{eq2:approxk-tree}), we get:
$$w(S')\geq w(E(\mathcal{S}))-w(E_{i_1}\cup E_{i_2})\geq \frac{k-1}{k+1}w(E(\mathcal{S}))\geq\frac{k-1}{k+1} \rho \cdot opt_{MaxWSSF}(G,w)$$
Finally, since $opt_{MaxWSSF}(G,w)\geq \uec(G,w)$ we get the expected result.
\qed
\end{proof}
\vspace{-0.4 cm}
\section{Approximation for bounded degree graphs}\label{sec:apx-MaxDegreeDelta}

In this section, we propose some positive approximation results for graphs of bounded degree in complement to those given in Corollary \ref{Maxdegree:BipReduc_IS}.

\begin{theorem}\label{theo:apx-MaxDegreeDelta}
In general graphs with maximum degree $\Delta$, there is an approximation preserving reduction from  \textsc{Weighted Upper Edge Cover} to  \textsc{MaxExtWSSF} with expansion $c(\rho)=\frac{1}{\Delta}\cdot \rho$.
\end{theorem}

\begin{proof}
Consider an edge-weighted graph $(G,w)$ of maximum degree $\Delta(G)$ bounded by $\Delta$ as an instance of \textsc{Weighted Upper Edge Cover}.
We make an instance $(G,w,U)$ of \textsc{MaxExtWSSF} by putting all pendant edges of $G$ in the forced edge set $U$. Property \ref{niceSSF} also works in this context since $U$ is the set of pendant edges. In particular, we deduce $opt_{ExtWSSF}(G,w,U)\geq \uec(G,w)$ because $U$ belongs to any minimal edge cover. Let $\mathcal{S}=\{S_1,\dots,S_r\}\subseteq E$ be a nice spanning star forest of $(G,w)$ containing $U$ satisfying:

\begin{equation}\label{eq1:apx-MaxDegreeDelta}
w(\mathcal{S})\geq \rho\cdot opt_{ExtWSSF}(G,w,U)\geq \rho\cdot \uec(G,w)
\end{equation}

For each $t\in\mathrm{Triv}$, we choose two edges incident to it with maximum weights $e_1^t=tx_t$ and $e_2^t=ty_t$ in $E\setminus E(\mathcal{S})$ (since by construction $d_G(v)\geq 2$), i.e.,  $w(e_1^t)\geq w(e_2^t)\geq w(tv)$ for all possible $v$; let $W=\sum_{t\in\mathrm{Triv}}\left(w(e_1^t)+w(e_2^t)\right)$ be this global quantity. Also, recall that $V_c$ and $V_l$ are the set of vertices labeled by centers and leaves respectively according to $\mathcal{S}$. We build a new vertex weighted graph $G(\mathcal{S})=G^\prime=(V^\prime,E^\prime,w^\prime)$ with maximum degree $\Delta(G^\prime)\leq \Delta(G)-1$ as follows:

\begin{itemize}
\item[$\bullet$] $V^\prime=V_l$.

\item[$\bullet$] $uv\in E^\prime$ iff there exists $t\in \mathrm{Triv}$ with $tx_t=tu$ and $ty_t=tv$.

\item[$\bullet$] For  $v\in V^\prime$, we set $w^\prime(v)=w\left(e_v(\mathcal{S})\right)$\footnote{We recall $e_v(\mathcal{S})$ is the edge of $\mathcal{S}$ linking leaf $v$ to its center.}.
\end{itemize}

Clearly, $G^\prime$ is a graph with bounded degree $\Delta-1$. We mainly prove that from any independent set $I\subseteq V^\prime$
we can polynomially build an upper edge cover $S_I$ of $G$ satisfying:

\begin{equation}\label{eq2:apx-MaxDegreeDelta}
w(S_I)\geq w^\prime(I)+\left(W-\sum_{t\in\mathrm{Triv}}w(e_1^t)\right)\geq w^\prime(I)
\end{equation}

Let $I\subseteq V^\prime$ be maximal independent set of $G^\prime$. This implies $V^\prime\setminus I$ is a vertex cover of $G^\prime$. By construction
of $G^\prime$, for every $t\in\mathrm{Triv}$, at least one vertex $x_t$ or $y_t$ is not in $I$ (say $x_t$ in the worst case). Recall $e_{x_t}(\mathcal{S})$
is the edge of spanning star forest incident to $x_t$ (since $x_t\in V_l$). We will  iteratively apply the following procedure for all
$t\in\mathrm{Triv}$ to build $S_I$:

if the current $\ell$-star $S_r$ of $\mathcal{S}$ containing $e_{x_t}(\mathcal{S})$ satisfies $\ell\geq 2$ (it is true initially by hypothesis),
then delete edge $e_{x_t}(\mathcal{S})$ from $\mathcal{S}$, add edge $e_1^t$ and update spanning star forest $\mathcal{S}$. Otherwise, $\ell=1$
and only add $e_1^t$. At the end of the procedure, we get a minimal edge cover $S_I$ of $G$ satisfying inequality (\ref{eq2:apx-MaxDegreeDelta}).\\

Now, apply as solution of $I$ the greedy algorithm of \textsc{MaxIS} for $G^\prime$ taking, at each step, one vertex with maximum weight $w^\prime$ and by removing all the remaining neighbors of it. It is well known that we have:

\begin{equation}\label{eq3:apx-MaxDegreeDelta}
w^\prime(I)\geq \frac{w^\prime(V^\prime)}{\Delta(G^\prime)+1}\geq \frac{w(\mathcal{S})}{\Delta(G)}
\end{equation}

Hence, using inequalities (\ref{eq1:apx-MaxDegreeDelta}), (\ref{eq2:apx-MaxDegreeDelta})  and  (\ref{eq3:apx-MaxDegreeDelta}),
we get the expected result.

\end{proof}

Using the $0.5$-approximation of \textsc{MaxExtWSSF} given in \cite{KhoshkhahGMT17}, we deduce:

\begin{corollary}\label{cor:apx-MaxDegreeDelta}
\textsc{Weighted Upper Edge Cover} is $\frac{1}{2\Delta}$-approximable in graphs with bounded degree $\Delta$.
\end{corollary}

\vspace{-0.4 cm}
\section{Conclusion}\label{sec:conclution}
In this article we gave positive and negative approximability aspects of \textsc{Weighted Upper Edge Cover} for special classes of graphs. We considered different types of weight function $w$ for edges of input graph.
Hardness of approximation on complete graphs when $w$ satisfies cycle inequality remains open.
Also for graphs with bounded degree $\Delta$, we have shown that our problem is $\frac{1}{2\Delta}$-approximable while we proved it can not be better than $\Theta\left(\frac{1}{\Delta}\right)$.
Finding a tighter approximation algorithm depending on $\Delta$ or on the average degree can be interesting.

\bibliographystyle{abbrv}
\bibliography{wuec}

\begin{thebibliography}{10}

\bibitem{AbouEishaHLMRZ16}
H.~AbouEisha, S.~Hussain, V.~V. Lozin, J.~Monnot, B.~Ries, and V.~Zamaraev.
\newblock A boundary property for upper domination.
\newblock In V.~M{\"{a}}kinen, S.~J. Puglisi, and L.~Salmela, editors, {\em
  Proc. of {IWOCA} 2016}, volume 9843 of {\em LNCS}, pages 229--240. Springer,
  2016.

\bibitem{AlimontiK00}
P.~Alimonti and V.~Kann.
\newblock Some {APX}-completeness results for cubic graphs.
\newblock {\em Theor. Comput. Sci.}, 237(1-2):123--134, 2000.

\bibitem{DBLP:conf/mfcs/AthanassopoulosCKK09}
S.~Athanassopoulos, I.~Caragiannis, C.~Kaklamanis, and M.~Kyropoulou.
\newblock An improved approximation bound for spanning star forest and color
  saving.
\newblock In R.~Kr{\'{a}}lovic and D.~Niwinski, editors, {\em Proc. of 34th
  {MFCS}}, volume 5734 of {\em LNCS}, pages 90--101. Springer, 2009.

\bibitem{Bertossi84}
A.~A. Bertossi.
\newblock Dominating sets for split and bipartite graphs.
\newblock {\em Inf. Process. Lett.}, 19(1):37--40, 1984.

\bibitem{BoothJ82}
K.~S. Booth and J.~H. Johnson.
\newblock Dominating sets in chordal graphs.
\newblock {\em {SIAM} J. Comput.}, 11(1):191--199, 1982.

\bibitem{BoriaCP15}
N.~Boria, F.~D. Croce, and V.~T. Paschos.
\newblock On the max min vertex cover problem.
\newblock {\em Discrete Applied Mathematics}, 196:62--71, 2015.

\bibitem{BourgeoisCEP13}
N.~Bourgeois, F.~D. Croce, B.~Escoffier, and V.~T. Paschos.
\newblock Fast algorithms for min independent dominating set.
\newblock {\em Discrete Applied Mathematics}, 161(4-5):558--572, 2013.

\bibitem{BoyaciM17}
A.~Boyaci and J.~Monnot.
\newblock Weighted upper domination number.
\newblock {\em Electronic Notes in Discrete Mathematics}, 62:171--176, 2017.

\bibitem{ChakrabartyG10}
D.~Chakrabarty and G.~Goel.
\newblock On the approximability of budgeted allocations and improved lower
  bounds for submodular welfare maximization and {GAP}.
\newblock {\em {SIAM} J. Comput.}, 39(6):2189--2211, 2010.

\bibitem{ChalermsookLN13}
P.~Chalermsook, B.~Laekhanukit, and D.~Nanongkai.
\newblock Graph products revisited: Tight approximation hardness of induced
  matching, poset dimension and more.
\newblock In S.~Khanna, editor, {\em Proc. of the 24th {SODA}}, pages
  1557--1576. {SIAM}, 2013.

\bibitem{Chang04}
G.~J. Chang.
\newblock The weighted independent domination problem is {NP}-complete for
  chordal graphs.
\newblock {\em Discrete Applied Mathematics}, 143(1-3):351--352, 2004.

\bibitem{ChenENRRS13}
N.~Chen, R.~Engelberg, C.~T. Nguyen, P.~Raghavendra, A.~Rudra, and G.~Singh.
\newblock Improved approximation algorithms for the spanning star forest
  problem.
\newblock {\em Algorithmica}, 65(3):498--516, 2013.

\bibitem{ChlebikC06}
M.~Chleb{\'{\i}}k and J.~Chleb{\'{\i}}kov{\'{a}}.
\newblock Complexity of approximating bounded variants of optimization
  problems.
\newblock {\em Theor. Comput. Sci.}, 354(3):320--338, 2006.

\bibitem{CorneilSKeil87}
D.~G. Corneil and J.~M. Keil.
\newblock A dynamic programming approach to the dominating set problem on
  $k$-trees.
\newblock {\em SIAM Journal on Algebraic Discrete Methods}, 8(4):535--543,
  1987.

\bibitem{DamaschkeMK90}
P.~Damaschke, H.~M{\"{u}}ller, and D.~Kratsch.
\newblock Domination in convex and chordal bipartite graphs.
\newblock {\em Inf. Process. Lett.}, 36(5):231--236, 1990.

\bibitem{DehneFFPR06}
F.~K. H.~A. Dehne, M.~R. Fellows, H.~Fernau, E.~Prieto{-}Rodriguez, and F.~A.
  Rosamond.
\newblock {NONBLOCKER:} parameterized algorithmics for minimum dominating set.
\newblock In {\em Proc. of the 32nd {SOFSEM}}, volume 3831 of {\em LNCS}, pages
  237--245. Springer, 2006.

\bibitem{Farber82ORL}
M.~Farber.
\newblock Independent domination in chordal graphs.
\newblock {\em Operations Research Letters}, 4(1):134--138, 1982.

\bibitem{Farber84}
M.~Farber.
\newblock Domination, independent domination and duality in strongly chordal
  graphs.
\newblock {\em Discrete Appl. Math.}, 7:115--130, 1984.

\bibitem{GJ79}
M.~R. Garey and D.~S. Johnson.
\newblock {\em Computers and Intractability: A Guide to the Theory of
  NP-Completeness}.
\newblock W. H. Freeman \& Co., New York, NY, USA, 1979.

\bibitem{HeL13}
J.~He and H.~Liang.
\newblock Improved approximation for spanning star forest in dense graphs.
\newblock {\em J. Comb. Optim.}, 25(2):255--264, 2013.

\bibitem{KhoshkhahGMT17}
K.~Khoshkhah, M.~K. Ghadikolaei, J.~Monnot, and D.~O. Theis.
\newblock Extended spanning star forest problems.
\newblock In {\em Proc. of the 11th {COCOA}}, volume 10627 of {\em LNCS}, pages
  195--209. Springer, 2017.

\bibitem{LozinMMZ17}
V.~V. Lozin, D.~S. Malyshev, R.~Mosca, and V.~Zamaraev.
\newblock More results on weighted independent domination.
\newblock {\em Theor. Comput. Sci.}, 700:63--74, 2017.

\bibitem{manlove1999algorithmic}
D.~F. Manlove.
\newblock On the algorithmic complexity of twelve covering and independence
  parameters of graphs.
\newblock {\em Discrete Applied Mathematics}, 91(1-3):155--175, 1999.

\bibitem{NguyenSHSMZ08}
C.~T. Nguyen, J.~Shen, M.~Hou, L.~Sheng, W.~Miller, and L.~Zhang.
\newblock Approximating the spanning star forest problem and its application to
  genomic sequence alignment.
\newblock {\em {SIAM} J. Comput.}, 38(3):946--962, 2008.

\bibitem{Nguyen15}
V.~H. Nguyen.
\newblock The maximum weight spanning star forest problem on cactus graphs.
\newblock {\em Discrete Math., Alg. and Appl.}, 7(2), 2015.

\bibitem{Slater77}
P.~J. Slater.
\newblock Enclaveless sets and mk-systems.
\newblock {\em J. Res. Nat. Bur. Stand.}, 82(3):197--202, 1977.

\bibitem{Trevisan01}
L.~Trevisan.
\newblock Non-approximability results for optimization problems on bounded
  degree instances.
\newblock In {\em Proc. of the 33rd {STOC}}, pages 453--461. {ACM}, 2001.

\bibitem{yannakakis1980edge}
M.~Yannakakis and F.~Gavril.
\newblock Edge dominating sets in graphs.
\newblock {\em SIAM Journal on Applied Mathematics}, 38(3):364--372, 1980.

\bibitem{Zuckerman07}
D.~Zuckerman.
\newblock Linear degree extractors and the inapproximability of max clique and
  chromatic number.
\newblock {\em Theory of Computing}, 3(1):103--128, 2007.

\end{thebibliography}

\end{document}